\documentclass[11pt]{article}

\usepackage{fullpage}
\usepackage{amsfonts, amssymb, amsmath, amsthm}
\usepackage{latexsym}
\usepackage[mathscr]{euscript}
\usepackage{url}
\usepackage{xcolor}
\definecolor{DarkGray}{rgb}{0.1,0.1,0.5}
\usepackage[colorlinks=true,breaklinks, linkcolor= DarkGray,citecolor= DarkGray,urlcolor= DarkGray]{hyperref}	
\usepackage{graphicx}
\usepackage[tight]{subfigure}
\usepackage{cancel}
\usepackage{multirow}
\usepackage{caption}	

\usepackage{import}
\usepackage{ltxtable}

\usepackage{calc}	

\usepackage[tracking=smallcaps]{microtype}	

\newcommand{\bra}[1]{{\langle#1|}}
\newcommand{\ket}[1]{{|#1\rangle}}
\newcommand{\braket}[2]{{\langle#1|#2\rangle}}

\newcommand{\proj}[1]{| #1 \rangle\!\langle #1 |}
\newcommand{\abs}[1]{{\lvert #1\rvert}}	
\newcommand{\norm}[1]{{\lVert #1 \rVert}}
\newcommand{\trnorm}[1]{{\lVert #1 \rVert_{\mathrm{tr}}}}
\newcommand{\Tr}{\mathrm{Tr}}

\DeclareMathOperator{\Ex}{\operatorname{E}}

\def\C {{\bf C}}

\def\D {{\mathcal D}}
\def\F {{\mathcal F}}
\def\H {{\mathcal H}}

\def\P {{\mathcal P}}

\def\N {{\bf N}}
\def\R {{\bf R}}

\newcommand{\MADV}{\mathrm{Madv}}

\newcommand{\evector}{v}

\newcommand{\ADV} {\mathrm{Adv}}

\newcommand{\blah} {\mathrm{Adv}^*}

\newcommand{\B}{\{0,1\}}	

\newcounter{sprows}

\newlength{\spheight}
\newlength{\spraise}

\newlength{\commentslength}

\newcommand{\rem}[1]{}
\newcommand{\ignore}[1]{}

\newtheorem{theorem}{Theorem}[section]
\newtheorem{lemma}[theorem]{Lemma}
\newtheorem{corollary}[theorem]{Corollary}
\newtheorem{claim}[theorem]{Claim}

\newtheorem{definition}[theorem]{Definition}
\newtheorem{remark}[theorem]{Remark}

\newfont{\subsubsecfnt}{ptmri8t at 10pt}
\renewcommand{\subparagraph}[1]{\smallskip{\subsubsecfnt #1.}}

\numberwithin{equation}{section} 

\newcommand{\eqnref}[1]{\hyperref[#1]{{(\ref*{#1})}}}
\newcommand{\thmref}[1]{\hyperref[#1]{{Theorem~\ref*{#1}}}}
\newcommand{\lemref}[1]{\hyperref[#1]{{Lemma~\ref*{#1}}}}
\newcommand{\corref}[1]{\hyperref[#1]{{Corollary~\ref*{#1}}}}
\newcommand{\defref}[1]{\hyperref[#1]{{Definition~\ref*{#1}}}}
\newcommand{\secref}[1]{\hyperref[#1]{{Section~\ref*{#1}}}}
\newcommand{\figref}[1]{\hyperref[#1]{{Figure~\ref*{#1}}}}
\newcommand{\tabref}[1]{\hyperref[#1]{{Table~\ref*{#1}}}}
\newcommand{\remref}[1]{\hyperref[#1]{{Remark~\ref*{#1}}}}
\newcommand{\appref}[1]{\hyperref[#1]{{Appendix~\ref*{#1}}}}
\newcommand{\claimref}[1]{\hyperref[#1]{{Claim~\ref*{#1}}}}
\newcommand{\propref}[1]{\hyperref[#1]{{Proposition~\ref*{#1}}}}
\newcommand{\exampleref}[1]{\hyperref[#1]{{Example~\ref*{#1}}}}
\newcommand{\conjref}[1]{\hyperref[#1]{{Conjecture~\ref*{#1}}}}

\allowdisplaybreaks[1]

\begin{document}
\title{A strong direct product theorem for quantum query complexity}
\author{%
Troy Lee \thanks{Centre for Quantum Technologies}%
\and%
 J\'{e}r\'{e}mie Roland \thanks{NEC Laboratories America}}

\maketitle

\begin{abstract}
We show that quantum query complexity satisfies a strong direct product theorem.  This 
means that computing $k$ copies of a function with less than $k$ times the 
quantum queries needed to compute one copy of the function implies that the overall success 
probability will be exponentially small in $k$.  For a boolean function $f$ we also show an XOR 
lemma---computing the parity of $k$ copies of $f$ with less than $k$ times the 
queries needed for one copy implies that the advantage over random guessing will be 
exponentially small.  

We do this by showing that the multiplicative adversary method, 
which inherently satisfies a strong direct product theorem, is always at least as large as the 
additive adversary method, which is known to characterize quantum query complexity.
\end{abstract}

\section{Introduction}
We show that quantum query complexity satisfies a strong direct product theorem.  A strong 
direct product theorem states that to compute $k$ copies of a function with less than $k$ times 
the resources needed to compute one copy of the function implies that the success probability will 
be exponentially small in $k$.  For boolean functions, we further show an XOR lemma.  
XOR lemmas are closely related to strong direct product theorems and state that  
computing the parity of $k$ copies of a boolean function with less than $k$ times 
the resources needed to compute one copy implies that the advantage over random guessing 
will be exponentially small.  XOR lemmas can be shown quite generally to imply strong direct 
product theorems and even threshold direct product theorems \cite{Unger09}, which state that 
one cannot compute a $\mu$ fraction of the $k$ copies with less than $\mu k$ times the 
resources with better than exponentially small (in $\mu k$) success probability.  Thus in the 
boolean case we are also able to obtain a threshold direct product theorem.   

How the resources needed to compute $k$ copies of a function scale with those needed for 
one copy is a very natural question that has been asked of many computational models.  
While direct product theorems are intuitively highly plausible, they do not hold in all models 
\cite{Sha03}, and there are relatively few models where strong direct product theorems are known.
Notable examples of direct product-type results include Yao's XOR lemma and 
Raz's parallel repetition theorem \cite{Raz98}.  Closer to our setting, strong direct 
product theorems have been shown for one-way randomized communication complexity 
\cite{Jain10} and for randomized query complexity \cite{Dru11}.  

In quantum query complexity strong direct product theorems were previously known for 
some special classes of functions and bounds shown by particular methods.  In the first 
such result, 
Klauck, \v{S}palek and de Wolf~\cite{KSW07} used the polynomial 
method \cite{BBC+98} to show a strong direct product theorem for the quantum query complexity 
of the OR function.  Via block sensitivity, this gives a polynomially tight strong direct product theorem for all functions---namely, any algorithm using less than a constant fraction times 
$k Q(f)^{1/6}$ will have exponentially small success probability for computing $k$ copies of $f$.  

Sherstov~\cite{She10} recently showed how certain 
lower bound techniques based on looking at the distance of the function to a convex set 
inherently satisfy a strong direct product theorem.  As an application he was able to show that 
the polynomial method satisfies a strong direct product theorem {\em in general}.  Thus one 
obtains a strong direct product theorem for the quantum query complexity of any function 
where the polynomial method shows a tight lower bound. 
Super-linear gaps between the polynomial degree and quantum query complexity are known 
\cite{Amb06}, however, so this does not give a tight strong direct product theorem for all functions.  

Direct product results have also been shown by the other main lower bound technique in quantum 
query complexity, the adversary method.  The adversary method defines a potential function 
based on the state of the algorithm after $t$ queries, and bounds the change in this potential 
function from one query to the next.  
By developing a new kind of adversary method, 
Ambainis, \v{S}palek, and de Wolf~\cite{ASW06} showed a strong direct product theorem for all 
symmetric functions.  \v{S}palek~\cite{Spa08} formalized this technique into 
a generic method, coining it the multiplicative adversary method, and showed that this method 
inherently satisfies a strong direct product theorem.  The name multiplicative adversary 
contrasts with the additive adversary method, introduced earlier by Ambainis~\cite{Amb02} and 
later extended by H{\o}yer, Lee and {\v{S}}palek~\cite{HLS07}.  The additive adversary method 
bounds the difference of the potential function from one step to the next, while the 
multiplicative adversary method bounds the corresponding ratio.  

There have recently been great strides in our understanding of the adversary methods.  
A series of works~\cite{FGG08,CCJY09,ACRSZ10,RS08,Rei09,Rei10,LMRS10} has 
culminated in showing that the additive adversary method characterizes the bounded-error 
quantum query complexity of any function whatsoever.  
Ambainis \textit{et al.}~\cite{AMRR11}, answering an open question of  
\v{S}palek~\cite{Spa08}, showed that the multiplicative adversary is at least as large as the 
additive.  Thus the multiplicative adversary bound also characterizes bounded-error 
quantum query complexity.  

This seems like it would close the question of a strong direct product theorem for quantum 
query complexity.  The catch is the following.  The multiplicative adversary method can be viewed as a family of methods parameterized by the bound $c$ on the ratio of the potential function from one step to the next.  The strong direct product theorem of~\cite{Spa08} holds for any value of $c$ sufficiently bounded away from $1$.  The result of~\cite{AMRR11}, however, was shown in the limit $c \rightarrow 1$, which ends up degrading the resulting direct product theorem into a direct sum theorem.  
We show that the multiplicative adversary is at least as large as 
the additive adversary for a value of $c$ bounded away from $1$.  A similar 
result was independently observed by Belovs \cite{Belovs11}.  Together with the strong direct product theorem 
for the multiplicative adversary by \cite{Spa08} this suffices to give a strong direct product theorem 
for quantum query complexity.  Rather than use this ``out of the box'' strong direct product theorem,
however, we prove the strong direct product theorem from scratch using a stronger output 
condition than those used previously \cite{Spa08, AMRR11}.  This results in better parameters,
and a better understanding of the multiplicative adversary method.

\begin{theorem}[Strong direct product theorem]
\label{thm:SDPT-main}
Let $f: \D \rightarrow E$ where $\D \subseteq D^n$ for finite sets $D,E$.  
For an integer $k > 0$ define $f^{(k)}(x^1, \ldots, x^k)=(f(x^1), \ldots, f(x^k))$. 
Then, for any $(2/3)\leq\delta\leq 1$,
\begin{align*}
Q_{1-\delta^{k/2}}(f^{(k)})
\ge \frac{k\ln(3\delta/2)}{8000}\cdot Q_{1/4}(f) 
\enspace .
\end{align*}
\end{theorem}

In the boolean case, we prove the following XOR lemma which also implies a threshold direct 
product theorem (\thmref{tdpt}).
\begin{lemma}[XOR Lemma]
\label{lem:XOR}
Let $f: \D \rightarrow \B$ where $\D \subseteq D^n$ for finite set $D$.  
For an integer $k > 0$ and any $0\leq\delta\leq 1$,
\begin{align*}
Q_{(1-\delta^{k/2})/2}(\oplus\circ f^{(k)})
\ge \frac{k\delta}{8000}\cdot Q_{1/4}(f) 
\enspace .
\end{align*}
\end{lemma}


\subsection{Proof technique}
While the statement of our main theorems concern functions, a key to our proofs, especially 
for the XOR lemma, is to 
consider more general state generation problems, introduced in~\cite{AMRR11}.  
Instead of producing a classical value $f(x)$ on input $x$, the goal in state generation is to produce a 
specified target state $\ket{\sigma_x}$, again by making queries to the input $x$.
We will refer to $\sigma(x,y)=\braket{\sigma_x}{\sigma_y}$ as the target Gram matrix.
Evaluating a function $f$ can be viewed as a special case of state generation where the 
target Gram matrix is $F(x,y)=\delta_{f(x),f(y)}$.  

Our most general result (\thmref{thm:product_sigma}) shows that for a restricted class of target 
Gram matrices $\sigma$, to generate $\sigma^{\otimes k}$ with better than exponentially 
small success probability requires at least a constant fraction of 
$k$ times the complexity of $\sigma$. The strong direct product theorem is obtained as a special case of this theorem by considering the Gram matrix $F(x,y)=\delta_{f(x),f(y)}$.
To obtain the XOR lemma, we apply this theorem with 
the state generation problem of computing $f$ in the phase, that is to generate 
$\sigma_f(x,y)=(-1)^{f(x)+f(y)}$.  The advantage of considering this state is that 
$\sigma_f^{\otimes k}$ is the state generation problem corresponding to computing 
the parity of $k$ copies of $f$ in the phase.  We then show that the complexities of $f$ and 
the state generation problem of computing $f$ in the phase are closely related.  

Another key element of our proofs is a new characterization of the set of valid output 
Gram matrices for an algorithm solving a state generation problem with success probability
 $1-\epsilon$ (\claimref{claim:fidelity}).  We call a condition which defines a set containing this set 
 of valid output matrices an output condition.  Usually a lower bound uses an 
 output condition which is a relaxation of the true output condition, and shows a lower bound 
 against all Gram matrices satisfying this output condition, and thereby all valid output 
 matrices as well.  Examples of output conditions 
 previously used with the adversary bound include being close to the target Gram matrix in 
 distance measured by the $l_\infty$ or $\gamma_2$ norms.  These conditions, however, 
 do not work for small success probabilities, which is critical to obtain the strong direct product 
 theorem.
 
We give a new characterization of the true output condition in terms of fidelity.  Since the fidelity between two quantum states is bounded by the fidelity between the probability distributions arising from any measurement on those states, a relaxation of this output condition may be obtained by considering the measurement corresponding to an optimal witness for the adversary bound of the problem.  A lower bound on the multiplicative bound under this relaxed output 
condition can be written as a linear program.  By taking the dual of this linear program we are 
able to lower bound the value on $\sigma^{\otimes k}$ in terms of the bound for $\sigma$ 
by using a completely classical claim about product probability distributions (\corref{cor:expectation-fidelity}). 
This approach allows us to obtain a cleaner statement for the strong direct product theorem than what we would obtain from the output condition used in~\cite{Spa08,AMRR11}, and also clarifies the inner workings of the adversary method, which might be of independent interest.

\section{Preliminaries}\label{sec:preliminaries}
Let $\Re(z)$ denote the real part of a complex number $z$.  
Let $\delta_{a,b}$ denote the Kronecker delta function.  We will refer throughout
to a function $f: \D \rightarrow E$ where $\D \subseteq D^n$ for finite sets $D,E$.  We let 
$f^{(k)}: \D^k \rightarrow E^k$ be the function computing $k$ independent copies of $f$, 
namely $f^{(k)}(x^1, \ldots, x^k)=(f(x^1), \ldots, f(x^k))$.  We let $\oplus \circ f^{(k)}$ denote 
the parity function composed with $f^{(k)}$.
We also define some auxiliary matrices associated with $f$.  Let $F(x,y)=\delta_{f(x),f(y)}$, 
and $\Delta_i(x,y)=\delta_{x_i,y_i}$ for $x,y \in \D$ and $i \in [n]$. For boolean functions, i.e., when $\abs{E}=2$, we also define the matrix $\sigma_f(x,y)=(-1)^{f(x)+f(y)}$ for $x,y\in\D$. Note that $\sigma_f=2F-J$, where $J$ is the all-$1$ matrix.  We use $A \circ B$ for 
the entrywise product between two matrices $A,B$, also known as the Schur or Hadamard product.

Let $\rho,\sigma$ be two $\abs{\D}\times\abs{\D}$ positive semidefinite matrices such that $\Tr\rho=\Tr\sigma=1$ (i.e., quantum states on a $\abs{\D}$-dim Hilbert space) and $p,q$ be two probability distributions over ${\D}$. We will use the notion of fidelity, for both quantum states and classical probability distributions.  
\begin{definition}[Fidelity]
\begin{align*}
 \F(\rho,\sigma)&=\Tr\sqrt{\sqrt{\rho}\sigma\sqrt{\rho}} &
 \F(p,q)&=\sum_{x\in\D}\sqrt{p_xq_x}
\end{align*}
\end{definition}

 For $0\leq \lambda\leq 1$ and $0< \mu< 1$, we denote by $D(\lambda||\mu)$ the binary relative entropy of $\lambda$ and $\mu$, defined as follows.
\begin{definition}[Binary relative entropy]
\[
 D(\lambda||\mu)=\lambda\ln\frac{\lambda}{\mu}+(1-\lambda)\ln\frac{1-\lambda}{1-\mu}
\]
where $0\ln 0=0$.
\end{definition}

Finally, for a $\abs{\D}\times\abs{\D}$ matrix $A$ we will also use the factorization norm $\gamma_2(A)$.
\begin{definition}[Factorization norm]
\begin{align*}
 \gamma_2(A)&=\min_{\substack{m\in\N\\ \ket{u_x},\ket{v_x}\in\C^m}}
\left\{\max_{x\in\D}\max\left\{\norm{\ket{u_x}}^2,\norm{\ket{v_x}}^2\right\}:
\forall x,y\in\D,A_{x,y}=\braket{u_x}{v_y}
\right\} \\
&=\max_{\substack{\ket{u},\ket{v} \\ \norm{\ket{u}}=\norm{\ket{v}}=1}} \trnorm{A \circ \ket{u}\bra{v}}
\end{align*}

\end{definition}

We will make use of the following basic claims.
\begin{claim}\label{claim:hadamard-product-properties}
For any matrices $A,B$ where $A \circ B$ is defined,
\begin{enumerate}
 \item $\norm{A\circ B}\leq\gamma_2(A)\cdot\norm{B}$
 \item $\left\{A\succeq 0\ \mathrm{and}\ B\succeq 0 \right\} \Rightarrow\ A\circ B \succeq 0$
\end{enumerate}
\end{claim}

%

\subsection{Quantum query complexity and state generation}
The quantum query complexity of $f$, denoted $Q_\epsilon(f)$ is the minimum number of 
input queries needed to compute $f$ with error at most $\epsilon$.  We refer to the 
survey \cite{BuhrmanDeWolf02querysurvey} for definitions and background on this model.

Although our main interest will be in the query complexity of functions, it will be useful 
to also talk about state generation problems, introduced in \cite{AMRR11}.  Instead of 
producing a classical value $f(x)$ on input $x$, the goal in state generation is to produce a 
specified target state $\ket{\sigma_x}$, again by making queries to the input $x$.  As unitary transformations independent of the input can be made for free 
in the query model, a state generation problem is wholly determined by the Gram matrix 
$\sigma(x,y)=\braket{\sigma_x}{\sigma_y}$ of the target states $\{\ket{\sigma_x}\}_{x \in \D}$.  
We refer to $\sigma$ as the target Gram matrix.

State generation problems come in two variations, coherent and non-coherent.  
An algorithm $\P$ solves the coherent quantum state generation problem $\sigma$ with error at 
most $\epsilon$ if, for every $x\in\D$, it generates a state $\ket{\P(x)}\in\H\otimes\H'$ such that 
$\Re(\bra{\P(x)} (\ket{\sigma_x}\otimes\ket{\bar{0}})) \geq \sqrt{1-\epsilon}$, where $\H'$ denotes 
the workspace of the algorithm, and $\ket{\bar{0}}$ is a default state for $\H'$. The 
coherent quantum query complexity of $\sigma$, denoted $Q_\epsilon^c(\sigma)$ is the minimum 
number of queries needed to generate $\sigma$ coherently with error at most $\epsilon$.  

An algorithm $\P$ solves the non-coherent state generation problem $\sigma$ with error at most 
$\epsilon$ if
 there exists a set of states $\ket{\phi_x}\in\H'$ such that 
$\Re(\bra{\P(x)} (\ket{\sigma_x}\otimes\ket{\phi_x}))\geq \sqrt{1-\epsilon}$ for all $x \in \D$.  
We denote by $Q_\epsilon(\sigma)$ the non-coherent query complexity of generating $\sigma$ 
with error $\epsilon$.  

Evaluating a function $f$ can be seen as a special case of non-coherent state generation where 
the target Gram matrix is $F(x,y)=\delta_{f(x),f(y)}$.  In other words, 
$Q_\epsilon(f)=Q_\epsilon(F)$ where $F(x,y)=\delta_{f(x),f(y)}$, justifying our 
abuse of notation.  For state generation problems corresponding to functions the coherent 
and non-coherent complexities are closely related.
\begin{claim}\label{claim:coherently}
Let $f$ be a function. Then
\[
Q_\epsilon(F)\leq Q_\epsilon^c(F)\leq 2Q_{1-\sqrt{1-\epsilon}}(F) \enspace .
\]
\end{claim}
 
\begin{proof}
The lower bound holds for a general target Gram matrix $\sigma$, as the success condition 
in the coherent case implies the non-coherent one.  

For the upper bound, let $A_x$ be an algorithm computing $f(x)$ with success probability 
$1-\eta$.
Thus the algorithm applied on $\ket{0}\ket{\bar{0}}$, where the first register is the output register and the second register corresponds to some workspace initialized in a default state, prepares a state
\[
 A_x \ket{0}\ket{\bar{0}}=\sum_{j} \alpha_j \ket{j+f(x)}\ket{\psi_j},
\]
where by assumption 
$|\alpha_{0}| \ge \sqrt{1-\eta}$, and the states $\ket{\psi_j}$ describe the final state of the workspace register.
Let us now copy the output register into an additional register initialized in the state $\ket{0}$ using an addition gate $G$, and finally uncompute the original output register together with the workspace by using the algorithm $A_x$ in reverse.

We analyze the overlap of $A_x^{-1} G A_x \ket{0}\ket{\bar{0}}\ket{0}$ with $\ket{0}\ket{\bar{0}}\ket{f(x)}$.  After applying $G$ on $A_x \ket{0}\ket{\bar{0}}\ket{0}$, we have the state 
$\ket{v}=\sum_{j} \alpha_j \ket{j+f(x)}\ket{\psi_j}\ket{j+f(x)}$.  Now we look at the overlap of
$\ket{0}\ket{\bar{0}}\ket{f(x)}$ with $A_x^{-1}\ket{v}$ or, equivalently, the overlap of $A_x\ket{0}\ket{\bar{0}}\ket{f(x)}$ with $\ket{v}$. Since
\[
 A_x\ket{0}\ket{\bar{0}}\ket{f(x)}=\sum_{j} \alpha_j \ket{j+f(x)}\ket{\psi_j}\ket{f(x)},
\]
we have
\[
 \bra{0}\bra{\bar{0}}\bra{f(x)}A_x^{-1}\ket{v}
=\sum_{j}  \abs{\alpha_j}^2\braket{f(x)}{j+f(x)}
\geq 1-\eta.
\]
Therefore, this algorithm coherently computes $f(x)$ with success probability $1-\epsilon\geq (1-\eta)^2$. Inverting this relation, we obtain $\eta\geq 1-\sqrt{1-\epsilon}$.
\end{proof}

We will also consider another type of state generation problem associated with a function, 
that of computing the function in the phase.  For a boolean function $f:\D \rightarrow \B$ let 
$\sigma_f(x,y)=(-1)^{f(x)+f(y)}$.  While the non-coherent complexity of $\sigma_f$ is trivial, 
the coherent complexity of $\sigma_f$ is closely related to that of $F$.

\begin{claim}\label{claim:computing-phase}
 \begin{align*}
  Q_{({1-\sqrt{1-\epsilon}})/{2}+{\epsilon}/{4}}^c(F)\leq Q_\epsilon^c(\sigma_f)\leq 2Q_{({1-\sqrt{1-\epsilon})}/{2}}(F)
 \end{align*}
\end{claim}
\begin{proof}
 For the lower bound, we turn an algorithm for $\sigma_f$ into an algorithm for $F=(J+\sigma_f)/2$ by using the SWAP test. The error dependence then follows from the joint concavity of the fidelity:
\begin{align*}
 \F\left(\tfrac{J+\rho}{2}\circ uu^*,\tfrac{J+\sigma_f}{2}\circ uu^*\right)\geq \frac{1}{2}+\frac{1}{2}\F\left(\rho\circ uu^*,\sigma_f\circ uu^*\right).
\end{align*}
for any $u$.

For the upper bound, let us consider an algorithm $A_x$ computing $f(x)$ (in a register) with success probability $1-\eta$.
Thus, the algorithm applied on $\ket{0}\ket{\bar{0}}$, where the first register is the output register and the second register corresponds to some workspace initialized in a default state, prepares a state
\[
 A_x \ket{0}\ket{\bar{0}}=\sum_{j=0,1} \alpha_j \ket{j\oplus f(x)}\ket{\psi_j},
\]
where by assumption 
$|\alpha_{0}| \ge \sqrt{1-\eta}$, and the states $\ket{\psi_j}$ describe the final state of the workspace register.
Let $\Phi$ be a phase gate acting on the output register as $\ket{b}\mapsto(-1)^{f(x)}\ket{b}$. We can turn an algorithm $A_x$ computing in a register into an algorithm computing in the phase by first applying $A_x$ to compute the output, then applying the phase gate $\Phi$, and finally applying $A_x^{-1}$ to uncompute the output.

After applying $\Phi$ on $A_x \ket{0}\ket{\bar{0}}$, we have the state 
$\Phi A_x\ket{0}\ket{\bar{0}}=\sum_{j=0,1} (-1)^{j+f(x)} \alpha_j \ket{j\oplus f(x)}\ket{\psi_j}$.  Now we look at the overlap of
$(-1)^{f(x)}\ket{0}\ket{\bar{0}}$ with $A_x^{-1}\Phi A_x \ket{0}\ket{\bar{0}}$ or, equivalently, the overlap of $(-1)^{f(x)}A_x\ket{0}\ket{\bar{0}}$ with $\Phi A_x \ket{0}\ket{\bar{0}}$. We have
\[
 (-1)^{f(x)}\bra{0}\bra{\bar{0}}A_x^{-1}\Phi A_x \ket{0}\ket{\bar{0}}
=\sum_{j} (-1)^j \abs{\alpha_j}^2
\geq 1-2\eta.
\]
Therefore we obtain a success probability $1-\epsilon\geq (1-2\eta)^2$. Inverting this relation, we obtain $\eta\geq(1-\sqrt{1-\epsilon})/2$.
\end{proof}

\section{Adversary methods}
In this section we introduce both the additive and multiplicative adversary lower 
bound methods.  Even when one is only interested in the functional case, it is useful 
to view these methods as lower bounds on quantum state generation as this allows 
the separation of the method into two distinct parts.  The first part is a lower bound on 
exact coherent quantum state generation.  This is where the two methods differ.  The second 
part is the output condition, a minimization of the bound for exact coherent quantum state 
generation over all valid output Gram matrices.  The set of valid output Gram matrices is 
determined by the target Gram matrix $\sigma$, the error parameter $\epsilon$, and if one is 
considering coherent or non-coherent state generation.  This second step is common to both 
the additive and multiplicative methods.  Finally, we show that the multiplicative bound is at least 
as large as the additive bound.  

\subsection{Additive method}
We first review the derivation of the additive adversary method to compare it with the 
multiplicative method in the next section.  We will actually present a generalization of the 
additive adversary method due to \cite{LMRSS11}.  

Consider an algorithm that exactly coherently 
computes $\sigma$ by making $T$ queries.  Let $\ket{\psi_x^t}$ be the state of this algorithm 
on input $x$ after $t$ queries, and $\rho^t(x,y)=\braket{\psi_x^t}{\psi_y^t}$ be the corresponding 
Gram matrix.  Note that $\rho^0=J$ the all ones matrix and, by assumption, $\rho^T=\sigma$.  

Now let $\Gamma$ be a matrix, $\evector$ a vector, and consider the potential function 
$\Phi(t)=\Tr((\Gamma \circ \rho^t) \evector\evector^*)$.  The additive change in this 
potential function from the beginning to the end of the protocol is 
\begin{align*}
\Tr((\Gamma \circ (J-\sigma)) \evector\evector^*)&= \sum_{t=0}^{T-1} 
\Tr((\Gamma \circ (\rho^t - \rho^{t+1})) \evector\evector^*) \\
&\le T \max_t \Tr((\Gamma \circ (\rho^t - \rho^{t+1})) \evector\evector^*) \enspace .
\end{align*}
A standard argument (see, for example, \cite{HLS07}) then goes that if we impose the 
condition on $\Gamma$ that
\[
I \pm \Gamma \circ (J-\Delta_i) \succeq 0 \ \mathrm{ for\  all } \ i \in [n],\\ 
\]
then $\Tr((\Gamma \circ (\rho^t - \rho^{t+1})) \evector\evector^*) \le 2$, for all $t$ and $\evector$.  

As this argument holds for any $\Gamma$ and $\evector$, we can maximize over 
them leading to the following definition.  
\begin{definition}[Additive adversary method~\cite{LMRSS11}]\label{def:additive}
\begin{equation*}
\begin{aligned}
\blah(\sigma) =&\underset{\Gamma}{\mathrm{maximize}}
& & \norm{\Gamma \circ (J-\sigma)} \\
& \mathrm{subject\  to}
& & I \pm \Gamma \circ (J-\Delta_i) \succeq 0 \ \mathrm{ for\  all } \ i \in [n],\\ 
\end{aligned}
\end{equation*}
where the maximization is over $\abs{\D}\times\abs{\D}$ hermitian matrices $\Gamma$.
\end{definition}
The preceding argument shows the following.
\begin{theorem}[\cite{LMRSS11}]
\label{thm:additive-tight}
For any target Gram matrix $\sigma$,
\[
Q_0^c(\sigma) \ge \frac{\blah(\sigma)}{2} 
\]
\end{theorem}

\cite{LMRSS11} have also shown that this lower bound is tight for the bounded-error query complexity of functions.
\begin{theorem}[\cite{LMRSS11}]
 For any function $f$,
\[
Q_{1/4}(f) \le 1000\cdot \blah(F) 
\]
\end{theorem}
Up to the constant factor, this upper bound holds more generally for \emph{well-behaved} state generation problems, where the query complexity $Q_\epsilon(\sigma)$ does not depend dramatically on the error $\epsilon$ (i.e., $Q_\epsilon(\sigma)=\Theta(Q_{\epsilon'}(\sigma))$ for small $\epsilon,\epsilon'$).

\begin{remark}
The adversary bound $\ADV^\pm$ from \cite{HLS07} was originally defined in the functional 
case, that is, for target Gram matrices $F$ of the form $F(x,y)=\delta_{f(x),f(y)}$ for a 
function $F$.  This definition had an additional constraint that $\Gamma \circ F=0$.  
This constraint only affects the bound up to a multiplicative factor of two 
\cite{LMRSS11}.  
\begin{equation}
\label{factoroftwo}
\ADV^\pm(F) \le \blah(F) \le 2 \ADV^\pm(F) \enspace .
\end{equation}
The constraint $\Gamma \circ F=0$ allows one to show that $\ADV^\pm(F)/2$ is a lower bound 
even on the non-coherent complexity of generating $F$.  One can see that $\blah(F)/4$ is a lower 
bound on the non-coherent complexity of generating $F$ either by Eq.~\eqref{factoroftwo} or by
\claimref{claim:coherently} showing that the coherent and non-coherent state generation 
complexities of functions are related by a factor of two.
\end{remark}

\subsection{Multiplicative adversary method}
The multiplicative bound is derived by considering the same potential function 
$\Phi(t)$, but looks at the ratio of this function 
at the beginning and end of the protocol, rather than the difference.  Equivalently, one can 
consider the logarithmic potential function $\ln(\Phi(t))$ and again look at the additive 
change over the course of the protocol.  
As the argument to the logarithm should be positive,
we already see that a new constraint on $\Gamma$ is needed, namely $\Gamma \succ 0$.  

\begin{definition}[Multiplicative adversary method]
\label{def:mult_exact}
\begin{equation*}
\begin{aligned}
\MADV(\sigma)=  \underset{c}{\mathrm{maximize}} 
\ \frac{1}{\ln(c)}\  &  \underset{\Gamma \succ 0,\evector}{\mathrm{maximize}}
& & \ln\left(\Tr((\Gamma \circ \sigma)\evector\evector^*) \right) \\
& \mathrm{subject\  to}
& & \Tr(\Gamma\evector\evector^*)=1 \\
& & & c^{-1}\Gamma\preceq \Gamma \circ \Delta_i\preceq c\ \Gamma \ \mathrm{ for\  all } \ i \in [n],\\ 
\end{aligned}
\end{equation*}
where the maximization is over $\abs{\D}\times\abs{\D}$ positive definite matrices $\Gamma$ and unit vectors $\evector$.
\end{definition}

\begin{theorem}[\cite{Spa08,AMRR11}]
 For any state generation problem $\sigma$,
\begin{align*}
Q^c_0(\sigma) &\ge \frac{\MADV(\sigma)}{2} \enspace .
\end{align*}
\end{theorem}

\begin{proof}
Consider an algorithm that coherently generates $\sigma$ by making $T$ queries, and define 
a potential function $\Phi(t)=\Tr((\Gamma \circ \rho^t) \evector\evector^*)$, where 
$\Gamma \succ 0$.  Then 
\begin{align*}
\frac{\Phi(T)}{\Phi(0)}&=
\frac{\Tr((\Gamma \circ \sigma)\evector\evector^*)}{\Tr((\Gamma \circ J)\evector\evector^*)}
=\prod_{t=0}^{T-1} 
\frac{\Tr((\Gamma \circ \rho^{t+1})\evector\evector^*)}{\Tr((\Gamma \circ \rho^{t})\evector\evector^*)} \\
&\le \left(\max_t \frac{\Tr((\Gamma \circ \rho^{t+1})\evector\evector^*)}{\Tr((\Gamma \circ \rho^{t})\evector\evector^*)}\right)^T.
\end{align*}

Analogously to the additive bound, we now show that the constraint 
$c^{-1}\Gamma\preceq \Gamma \circ \Delta_i\preceq c\ \Gamma \ \mathrm{ for\  all } \ i \in [n]$ 
implies 
\[
\max_t \frac{\Tr((\Gamma \circ \rho^{t+1})\evector\evector^*)}{\Tr((\Gamma \circ \rho^{t})\evector\evector^*)}
\le c \enspace .
\]
This argument is very similar to proofs in~\cite{Spa08,AMRR11} so we only sketch the idea here.
Recall from~\cite{AMRR11} that we can assume that there are only two types of queries, called computing and uncomputing queries (this restriction can only increase the query complexity by a factor at most $2$, hence the factor $1/2$ in the final lower bound).
Let us first consider a computing query.  Let $\ket{\psi_{x,i}^t}=P_i \ket{\psi_x^t}$, where $P_i$ is 
a projector onto the query register containing index $i$, and 
$\rho_i^t(x,y)=\braket{\psi_{x,i}^t}{\psi_{y,i}^t}$.  We can decompose the state before the $t$-th query as $\rho^t=\sum_i\rho_i^t$, and the state after the query as 
$\rho^{t+1}=\sum_i\rho_i^t\circ\Delta_i$. The condition 
$\Gamma \circ \Delta_i\preceq c\ \Gamma$ then immediately implies that
\begin{align*}
 \Tr((\Gamma \circ \rho^{t+1})\evector\evector^*)\leq c\ \Tr((\Gamma \circ \rho^{t})\evector\evector^*).
\end{align*}
For uncomputing queries, the roles of $\rho^t$ and $\rho^{t+1}$ are interchanged, and we obtain the same conclusion from the constraint $\Gamma\preceq c\ \Gamma \circ \Delta_i$.
\end{proof}

\begin{remark}
The constraints on $\Gamma$ given here are expressed differently from~\cite{Spa08, AMRR11}, 
the latter using the constraint  $\norm{\Gamma^{1/2} (\Gamma \circ \Delta_i)^{-1/2}}^2\leq c$ and
$\norm{(\Gamma \circ \Delta_i)^{1/2}\Gamma^{-1/2}}^2\leq c$.
It is straightforward to show, however, that these conditions are equivalent to 
$c^{-1}\Gamma\preceq \Gamma \circ \Delta_i\preceq c\ \Gamma$.
\end{remark}
%
%

When the value of $c$ is fixed, the multiplicative bound becomes a semidefinite program.
Indeed, setting $W=\Gamma\circ\evector\evector^*$, we have:
 \begin{equation*}
\begin{aligned}
\MADV(\sigma)=  \underset{c}{\mathrm{maximize}} 
\ \frac{1}{\ln(c)}\  &  \underset{W \succ 0}{\mathrm{maximize}}
& & \ln\left(\Tr(W \sigma)\right) \\
& \mathrm{subject\  to}
& & \Tr(WJ)=1 \\
& & & c^{-1}W\preceq W \circ \Delta_i\preceq c\ W \ \mathrm{ for\  all } \ i \in [n].
\end{aligned}
\end{equation*}
Thus we can view the 
multiplicative adversary bound as a maximization over semidefinite programs.

\subsection{Output condition}
Thus far, we have seen lower bounds on the problem of {\em exact coherent} state generation.  
To obtain a lower bound in the bounded-error setting---coherent or non-coherent---one can 
minimize the exact coherent bound over the set of valid final Gram matrices of a successful 
algorithm.  

We will restrict our discussion to the coherent output condition.  As our main results are for 
functions, by showing lower bounds on the coherent state generation problems 
$F$ and $\sigma_f$ associated with a function $f$, we obtain lower bounds on the 
query complexity of $f$ by \claimref{claim:coherently} and 
\claimref{claim:computing-phase}.

Recall that a successful coherent $\epsilon$-error algorithm $\P$
for the set of target vectors $\{\sigma_x\}$ must satisfy 
$\Re(\bra{\P(x)} (\ket{\sigma_x}\otimes\ket{\bar{0}})) \geq \sqrt{1-\epsilon}$.  We can 
equivalently rephrase this as 
$\Re(\bra{\P(x)} V(\ket{\sigma_x}\otimes\ket{\bar{0}})) \geq \sqrt{1-\epsilon}$ for some 
unitary $V$.  This can be done as the unitary $V$ can be appended to the algorithm at no 
extra cost, and this formulation has the advantage that it only depends on the Gram matrix 
$\sigma$ of the vectors $\{\sigma_x\}$ and the Gram matrix 
$\sigma'(x,y)=\braket{\P(x)}{\P(y)}$, rather than the vectors themselves.  

The set of $\sigma'$ satisfying this condition can be hard to deal with, so previous works 
have typically relaxed this condition and used an output condition that defines a larger, 
simpler set.
For example, the original Ambainis output condition minimized over $\sigma'$ satisfying 
$\ell_\infty(\sigma - \sigma') \le 2 \sqrt{\epsilon}$ for error parameter $\epsilon$.  
A stronger output condition based on the $\gamma_2$ norm that 
$\gamma_2(\sigma - \sigma') \le 2 \sqrt{\epsilon}$ was introduced in 
\cite{HLS07}.  As $\gamma_2(v) \ge \ell_\infty(v)$, this output condition defines a smaller set.  The $\gamma_2$ output condition was later shown to be approximately tight in the sense that if 
$\gamma_2(\sigma -\sigma') \le \epsilon$, then there is a unitary $V$ such that 
$\bra{\sigma_x}V\ket{\sigma_x'} \ge 1- 2\sqrt{\epsilon}$ for all $x$ \cite{LMRSS11}.  
While approximately tight in the bounded-error setting, this condition is not strong enough for 
proving strong direct product theorems, where we need to obtain non-trivial bounds for 
exponentially small success probabilities.

Here we work with the full output condition and express it in an alternative form that is easier 
to handle.  As a side effect, our new characterization provides an alternative proof that the  
$\gamma_2$ output condition is tight in the bounded-error setting, and improves the parameters 
given in \cite{LMRSS11}.

\begin{claim}\label{claim:fidelity}
Let $\{\ket{a_x}\},\{\ket{b_x}\}$ be two sets of vectors, and $\rho,\sigma$ their corresponding Gram matrices.
 \begin{equation}
\max_V\min_x \Re(\bra{a_x}V\ket{b_x})=
\min_{u:\norm{u}=1} \F(\rho\circ uu^*,\sigma\circ uu^*) \enspace,
\end{equation}
where the maximization is taken over all unitaries $V$.
\end{claim}

\begin{proof}
By writing the left hand side as a semidefinite program and taking the dual one can show that
\[
\max_V\min_x \Re(\bra{a_x}V\ket{b_x})=\min_{u: \norm{u}=1} \max_V 
\Re(\Tr(V \sum_x |u_x|^2 \ket{a_x}\bra{b_x})) \enspace .
\]

Letting $D(u)$ be a diagonal matrix with entries given by $u$, we can rewrite the right 
hand side of this last expression as
\begin{align*}
 \max_V \min_x \Re(\bra{a_x}V\ket{b_x})
&=\min_{u: \norm{u}=1} 
\trnorm{A D(u)  (BD(u))^*} \enspace,
\end{align*}
where $A=\sum_x \ket{a_x}\bra{x}$ and $B=\sum_x \ket{b_x}\bra{x}$.
Since $\rho=A^* A$, $\sigma=B^* B$ and $\rho\circ uu^*=D(u)^*\rho D(u)$, the claim follows using 
\[
\trnorm{XY^*}=\trnorm{(X^*X)^{1/2}(Y^*Y)^{1/2}}
\]
and the definition of the fidelity $\F(X^*X,Y^*Y)=\trnorm{(X^*X)^{1/2}(Y^* Y)^{1/2}}$.
\end{proof}

The following quantities then give lower bounds for $\epsilon$-error  
coherent quantum state generation:
\begin{definition}[Additive and multiplicative bounds]
\begin{align*}
 \ADV_\epsilon(\sigma)&=\min_\rho\blah(\rho)\\
 \MADV_\epsilon(\sigma)&=\min_\rho\MADV(\rho),
\end{align*}
where both minimizations are over Gram matrices $\rho$ such that
\begin{align*}
 \min_{u:\norm{u}=1} \F(\rho\circ uu^*,\sigma\circ uu^*)\ge \sqrt{1-\epsilon}.
\end{align*}
\end{definition}

In light of \claimref{claim:fidelity}, we can slightly improve one of the bounds in~\cite[Lemma 4.8]{LMRSS11}, which compares the tight output condition based on the fidelity to the output condition based on the factorization norm $\gamma_2$.
\begin{claim}
Let $\{\ket{a_x}\},\{\ket{b_x}\}$ be two sets of vectors, and $\rho,\sigma$ their corresponding Gram matrices.  Say that $\sqrt{1-\epsilon}=\max_V\min_x \Re(\bra{a_x}V\ket{b_x})$, where the maximization is taken over all unitary matrices $V$.
Then
\[
1-\sqrt{1-\epsilon} \le \frac{1}{2}\gamma_2(\rho-\sigma) \le \sqrt{\epsilon},
\]

\end{claim}
\begin{proof}
This directly follows from \claimref{claim:fidelity} and the relation between the trace distance 
and fidelity.
\[
1-\F(\rho\circ uu^*,\sigma\circ uu^*) \le \frac{1}{2}\trnorm{(\rho-\sigma)\circ uu^*} \le \sqrt{1-\F(\rho\circ uu^*,\sigma\circ uu^*)^2} \enspace . 
\]
\end{proof}

Note that an adversary matrix $\Gamma$ yields a good zero-error multiplicative adversary bound if $\Tr(\Gamma(\sigma\circ\evector\evector^*))$ is large. To obtain a bound for $\epsilon$-error 
algorithms, we need to show that $\Tr(\Gamma(\rho\circ\evector\evector^*))$ remains large for any Gram matrix $\rho$ such that $\F(\rho\circ uu^*,\sigma\circ uu^*)\ge \sqrt{1-\epsilon}$ for all unit vectors $u$. The following lemma will be useful.
\begin{lemma}\label{lem:expectation-fidelity}
Let $p,q$ be two distributions for a discrete random variable $A$ taking values in 
$\R^+_0$. If $\F(p,q)\geq\sqrt{\delta}$, then
\begin{align*}
 E_q(A)\geq \delta \left[E_p(A^{-1})\right]^{-1}.
\end{align*}
\end{lemma}
\begin{proof}
Let $p_i=\Pr_p[A=a_i]$ and $q_i=\Pr_q[A=a_i]$. We need to lower bound the value of the following optimization program:
\begin{align*}
 \begin{aligned}
\underset{q_i\geq 0:\ \sum_iq_i=1}{\mathrm{minimize}}\ \sum_i q_i a_i
& \ \mathrm{subject\  to}
& & \F(p,q)\geq\sqrt{\delta}.
\end{aligned}
\end{align*}
Introducing vectors $\ket{u}=\sum_i\sqrt{p_i}\ket{i}$ and $\ket{v}=\sum_i\sqrt{q_i}\ket{i}$, and 
letting $D(A)$ be a diagonal matrix with the support of $A$ along the diagonal, this can be 
rewritten as
\begin{align*}
 \begin{aligned}
\underset{\ket{v}: \norm{v}=1}{\mathrm{minimize}}\ \bra{v}D(A)\ket{v}
& \ \mathrm{subject\  to}
& & \abs{\braket{u}{v}}^2\geq\delta\\
=
\underset{\rho\succeq 0: \Tr\rho=1}{\mathrm{minimize}}\ \Tr[D(A)\rho]
& \ \mathrm{subject\  to}
& & \Tr[\proj{u}\rho]\geq\delta.
\end{aligned}
\end{align*}
This is a semidefinite program, whose dual can be written as
\begin{align*}
 \begin{aligned}
\underset{\lambda\geq 0,\mu}{\mathrm{maximize}}\ \lambda \delta+\mu
& \ \mathrm{subject\  to}
& & D(A)\succeq \lambda\proj{u}+\mu I.
\end{aligned}
\end{align*}
Setting $\mu=0$, this is at least
\begin{align*}
 \begin{aligned}
\delta\ \underset{\lambda\geq 0}{\mathrm{maximize}}\ \lambda 
& \ \mathrm{subject\  to}
& & D(A) \succeq \lambda\proj{u}.
\end{aligned}
\end{align*}
Let $\ket{w}=\sum_i \sqrt{p_i/a_i}\ket{i}$.  The constraint 
is equivalent to $I \succeq \lambda\proj{w}$, which in turn is equivalent to 
$\lambda \norm{\proj{w}} = \lambda \norm{w}^2 \le 1$. The lemma then follows from $\norm{w}^2=\sum_i p_i a_i^{-1}$.
\end{proof}

To apply this lemma, we need an upper bound on $E_p[A^{-1}]$.  In our applications, we 
usually do not know explicitly the distribution $p$, but we do know its expectation and the 
extremal values in its support.   The next claim allows us to upper bound 
$E_p[A^{-1}]$ in terms of these quantities.

\begin{claim}\label{claim:bound-expectation}
 Let $0<a_0\leq \bar{a}\leq a_1$, and $A$ be a random variable taking values in $S\subseteq[a_0,a_1]$. If $E_p(A)=\bar{a}$, then $E_p(A^{-1})\leq \tfrac{a_0+a_1-\bar{a}}{a_0a_1}$.
\end{claim}
\begin{proof}
 $E_p(A^{-1})$ is at most the value of the following linear program:
\begin{align*}
 \begin{aligned}
\underset{p_a\geq 0}{\mathrm{maximize}}\ \sum_{a\in S}p_a a^{-1}
& \ \mathrm{subject\  to}
& &\sum_{a\in S}p_a a =\bar a, \quad \sum_{a\in S}p_a =1.
\end{aligned}
\end{align*}
The dual program can be written as
\begin{align*}
 \begin{aligned}
\underset{\lambda,\mu}{\mathrm{minimize}}\ \lambda-\bar a \mu
& \ \mathrm{subject\  to}
& & \mu a^2-\lambda a+1\leq 0\ \forall a\in S.
\end{aligned}
\end{align*}
Since $a_0\leq a\leq a_1$, the constraint is satisfied for $\lambda=\tfrac{a_0+a_1}{a_0a_1}$ and $\mu=\tfrac{1}{a_0a_1}$, which leads to $E_p(A^{-1})\leq \tfrac{a_0+a_1-\bar{a}}{a_0a_1}$.
\end{proof}

Putting the last two claims together, we get the following corollary which is key to our strong 
direct product theorem.
\begin{corollary}\label{cor:expectation-fidelity}
Let $a_1\ge a_0 >0$ and $p$ be a distribution for a random variable $A$ taking values in $[a_0,a_1]$.  
If $E_p[A]=\bar a$ and $q$ is a distribution over $(\R^+_0)^k$ such that 
$\F(p^{\otimes k},q)\geq\sqrt{\delta^k}$, then
\begin{align*}
 E_q(\Pi_{l=1}^k A_l)\geq  \left(\frac{\delta a_0 a_1}{a_0+a_1-\bar a}\right)^k.
\end{align*}
\end{corollary}

\subsection{Comparison of the adversary bounds}

Let us first prove a variation of the result by~\cite{AMRR11} that the multiplicative adversary bound 
is stronger than the additive bound.  The main difference with~\cite{AMRR11} is that this claim 
relies on the bound $\blah(\sigma)$ which is potentially stronger for general quantum state 
generation problems.
\begin{claim}[\cite{AMRR11}]\label{claim:madvadv-limit}
For any $\epsilon>0$ and any state generation problem $\sigma$, we have
 \[
\MADV(\sigma) \ge (1-\epsilon) \blah(\sigma) .
\]
\end{claim}
\begin{proof}
Let $\Gamma$ be an optimal witness for $\blah(\sigma)=b$, and $\evector$ be the principal eigenvector of $\Gamma\circ(J-\sigma)$. Note that we may assume without loss of generality that $\evector$ corresponds to a positive eigenvalue of $\Gamma\circ(J-\sigma)$. Let $\Gamma'=\Gamma-\Tr((\Gamma\circ\sigma)\evector\evector^*) I$, and notice that $\Gamma'$ is also a witness for $\blah(\sigma)=b$, satisfying $\Tr(\Gamma'\evector\evector^*)=b$ and $\Tr((\Gamma'\circ\sigma)\evector\evector^*)=0$. Let $d=\norm{\Gamma'}$ and note that $d \ge b$. Finally, define $\Gamma_m=(I+\gamma(dI-\Gamma'))/(1+\gamma(d-b))$. Therefore, we have $\Tr(\Gamma_m \evector\evector^*)=1$ and $\Tr((\Gamma_m\circ\sigma)\evector\evector^*)=(1+\gamma d)/(1+\gamma(d-b))$.  

We now show that the condition 
$c^{-1}\Gamma_m\preceq \Gamma_m \circ \Delta_i\preceq c\Gamma_m$ is satisfied for $c=1+\gamma$.  We show $(1+ \gamma(d-b)) (\Gamma_m \circ (c\Delta_i -J)) \succeq 0$ which 
implies $\Gamma_m \circ (c\Delta_i -J) \succeq 0$ as $1+\gamma(d-b) >0$.  
\begin{align*}
(1+ \gamma(d-b)) (\Gamma_m \circ (c\Delta_i -J))&= 
\Big((1+\gamma d) I - \gamma \Gamma' \Big) \circ \Big((\Delta_i-J) + \gamma \Delta_i\Big)\\
&= \gamma(I+\Gamma' \circ (J-\Delta_i)) + \gamma^2 ( dI - \Gamma') \circ \Delta_i.
\end{align*}
From the constraint of the additive metric we know that 
$I + \Gamma' \circ (J-\Delta_i) \succeq 0$ for all $i \in [n]$.  Also as $dI - \Gamma' \succeq 0$ , taking 
the Hadamard product with $\Delta_i\succeq 0$ gives $(d I - \Gamma')\circ \Delta_i \succeq 0$, by 
Property~2 in \claimref{claim:hadamard-product-properties}. Therefore, we have 
$\Gamma_m \circ (c\Delta_i -J)\succeq 0$.  One can show 
$\Gamma_m \circ (cJ -\Delta_i)\succeq 0$ in a similar fashion. This implies that $\Gamma_m$ is a witness for
 \[
\MADV(\sigma) \ge \frac{\ln \left(\tfrac{1+\gamma d}{1+\gamma (d-b)}\right)}{\ln (1+\gamma)}.
\]
The right hand side tends to $b=\blah(\sigma)$ in the limit $\gamma\to 1$, therefore, by continuity, for any $\epsilon>0$ there exists $\gamma$ such that $\MADV(\sigma)\geq (1-\epsilon)\blah(\sigma)$.
\end{proof}
Adapting results from~\cite{Spa08,AMRR11}, this implies a strong direct product theorem for $\MADV(\sigma)$ as long as the bound is obtained for $c=1+\Omega(1/\blah(\sigma))$. Unfortunately, showing that we can take $c$ bounded away from $1$ requires bounding 
$d=\norm{\Gamma'}$, which we do not know how to do for a general state generation problem 
$\sigma$.   In general, we can only use this statement in the limit $c\to 1$, in which case the direct product theorem degrades into a direct sum theorem.  This is why \cite{AMRR11} were not 
able to conclude a strong direct product theorem.  

We observe that for interesting cases such as 
$F$ or $\sigma_f$, we {\em can} bound the norm of the witness $\Gamma'$ using the following 
claim.  

\begin{claim}\label{claim:witness}
Suppose that $(J-\sigma) \circ (J-\sigma)=\lambda (J-\sigma)$.  Then there is a matrix  
$\Gamma'$ witnessing $\lambda \frac{\blah(\sigma)}{\gamma_2(J-\sigma)}$ such that 
$\norm{\Gamma'} \le \tfrac{\blah(\sigma)}{\gamma_2(J-\sigma)}$ and $\Gamma'\circ (J-\sigma)=\lambda \Gamma'$.   
\end{claim}

\begin{proof}
Let $\Gamma$ be an optimal witness for $\blah(\sigma)$.  Define 
$\Gamma'=\gamma_2(J-\sigma)^{-1} (\Gamma \circ (J-\sigma))$. By assumption, we then have $\Gamma'\circ (J-\sigma)=\lambda \Gamma'$. This is a feasible witness as 
\[
\norm{\Gamma' \circ (J-\Delta_i)} \le 
\frac{\gamma_2(J-\sigma)}{\gamma_2(J-\sigma)} \norm{\Gamma \circ (J-\Delta_i)} \le 1
\]
by 
Property~1 in \claimref{claim:hadamard-product-properties}.
Furthermore, $\norm{\Gamma'}=\gamma_2(J-\sigma)^{-1}\blah(\sigma)$ and 
$\Gamma'$ witnesses a bound of 
$\lambda \norm{\Gamma'}= 
\lambda \gamma_2(J-\sigma)^{-1} \blah(\sigma)$. 
\end{proof}

For certain state generation problems including $F$ and $\sigma_f$ we are 
thus able to obtain a quantitative version of \claimref{claim:madvadv-limit}.

\begin{claim}
\label{claim:madvadv}
Suppose that $(J-\sigma) \circ (J-\sigma)=\lambda (J-\sigma)$, and let $d=\gamma_2(J-\sigma)^{-1}\blah(\sigma)$. Then, for any $\gamma>0$, there is a multiplicative witness $\Gamma_m$ and a vector $\evector$ such that
\begin{align*}
 \Tr(\Gamma_m\evector\evector^*)&=1\\
 \Tr(\Gamma_m(\sigma \circ\evector\evector^*))&= 1+\lambda \gamma d\\
I\preceq\Gamma_m&\preceq (1+2\gamma d)I,\\
c^{-1}\Gamma_m\preceq \Gamma_m \circ \Delta_i&\preceq c\ \Gamma_m \text{ for all } i,
\end{align*}
where $c=1+\gamma$.
Therefore $\Gamma_m$ satisfies the constraints of \defref{def:mult_exact} and witnesses that
\[
\MADV(\sigma) \ge \frac{\ln (1+\lambda \gamma d)}{\ln (1+\gamma)}
\]
\end{claim}

\begin{proof}
From Claim~\ref{claim:witness}, there exists a witness $\Gamma$ witnessing $\blah(\sigma)\geq\lambda d$ such that $\norm{\Gamma}=d$. Let $\evector$ be the principal eigenvector of $\Gamma$, and $\Gamma_m=I+\gamma(dI-\Gamma)$. Note that we may assume without loss of generality that $\evector$ corresponds to a positive eigenvalue of $\Gamma$. Therefore, we have $\Gamma_m\succeq I$ and  $\Tr(\Gamma_m \evector\evector^*)=1$.  
As $\Gamma \circ (J-\sigma)=\lambda \Gamma$, it follows that $\evector$ is also a principal
eigenvector of $\Gamma \circ (J - \sigma)$, and the objective value achieved by $\Gamma$ 
is $\Tr(\Gamma((J-\sigma) \circ \evector\evector^*))=\lambda d$.  Thus 
$\Tr(\Gamma (\sigma \circ \evector\evector^*))=(1-\lambda)d$ and 
$\Tr(\Gamma_m (\sigma \circ \evector\evector^*))=1+\lambda \gamma d$.    
The third condition follows from $-d I\preceq\Gamma\preceq d I$. 

The fact that the condition $c^{-1}\Gamma_m\preceq \Gamma_m \circ \Delta_i\preceq c\Gamma_m$ is satisfied for $c=1+\gamma$ follows by the same argument as in the proof of Claim~\ref{claim:madvadv-limit}.
%
\end{proof}

Taking $\gamma=1/(d\lambda)$ gives the following corollary.
\begin{corollary}
Suppose that $(J-\sigma) \circ (J-\sigma)=\lambda (J-\sigma)$. Then,
 \[
\MADV(\sigma) \ge \lambda \frac{\blah(\sigma)}{2} .
\]
\end{corollary}
Note that in this statement $\MADV(\sigma)$ is proved with $c=1+1/(\lambda\cdot \blah(\sigma))$, which is what we need for the strong direct product theorem.

Now we have shown that the multiplicative bound is a constant fraction of the additive bound 
in the exact case.  Thus the same will be true with respect to any 
output condition.

\section{Strong direct product theorem}
We first prove the following theorem, which will lead to both the strong direct product theorem and the XOR lemma in the boolean case.
\begin{theorem}
\label{thm:product_sigma}
Let $\sigma$ be a Gram matrix for a state generation problem satisfying $(J-\sigma) \circ (J-\sigma)=\lambda (J-\sigma)$ for some $\lambda>0$, and let $d=\gamma_2(J-\sigma)^{-1}\blah(\sigma)$.  Then for any 
$\gamma > 0$ 
\[
Q_{1-\delta^k}^c(\sigma^{\otimes k}) \geq 
\frac{k\ln \left(\delta \tfrac{1+2\gamma d}{1+\gamma d(2-\lambda)}\right)}{2\ln (1+\gamma)} \enspace .
\]
\end{theorem}

\begin{proof}
Let $\evector,\Gamma_m$ satisfy the conditions in \claimref{claim:madvadv}.  
As a witness for $\sigma^{\otimes k}$ we 
take $\Gamma_m^{\otimes k}$.  Let us first see that this matrix satisfies the multiplicative constraint 
with the same value $c=1+\gamma$.  

We label the constraint matrices $\Delta_{p,q}$ for $\sigma^{\otimes k}$ by 
$p \in [k]$ and $q \in [n]$.  These are $|\D|^k$-by-$|\D|^k$ matrices where 
$\Delta_{p,q}((x^1, \ldots, x^k), (y^1, \ldots, y^k))=\delta_{x^p_q,y^p_q}$.  In other words, 
$\Delta_{p,q}=J^{\otimes p-1} \otimes \Delta_q \otimes J^{\otimes k-p}$.  Thus 
$\Gamma^{\otimes k} \circ \Delta_{p,q}=\Gamma_m^{\otimes p-1} \otimes 
\Gamma_m \circ \Delta_q \otimes \Gamma_m^{\otimes k-p}$. 
Since $c^{-1}\Gamma_m \preceq \Gamma_m \circ \Delta_q \preceq c\ \Gamma_m$ for all 
$p \in [n]$, and obviously $c^{-1}\Gamma_m \preceq \Gamma_m \preceq c\ \Gamma_m$ for 
$c>1$, we immediately have
\begin{align*}
 c^{-1}\Gamma_m^{\otimes k} \preceq \Gamma_m^{\otimes k} \circ \Delta_{p,q}\preceq c\ 
 \Gamma_m^{\otimes k}
\end{align*}
for any $p \in [k], q \in [n]$.

To lower bound the objective value we lower bound
\[
\MADV_{1-\delta^k}(\sigma^{\otimes k}) \ge 
\min_\rho \Tr(\Gamma_m^{\otimes k} (\rho \circ (\evector\evector^*)^{\otimes k})),
\]
where the minimum is over psd matrices $\rho$ such that $\rho\circ I=I$ and
\begin{align*}
 \min_u \F(\rho\circ uu^*,\sigma^{\otimes k}\circ uu^*)\geq \delta^{k/2}.
\end{align*}
In particular, this will hold for $u=\evector^{\otimes k}$ and we can apply 
\corref{cor:expectation-fidelity} with $p$ being the distribution arising from measuring 
$\Gamma_m$ on $\sigma \circ \evector\evector^*$, and $q$ the distribution arising from 
measuring $\Gamma_m^{\otimes k}$ on $\rho \circ (\evector\evector^*)^{\otimes k}$. Since 
$\F(\rho\circ (\evector\evector^*)^{\otimes k},(\sigma\circ \evector\evector^*)^{\otimes k})\geq \delta^{k/2}$, we also have 
$\F(p^{\otimes k},q)\geq\delta^{k/2}$.  The parameters in \corref{cor:expectation-fidelity} are
$a_0=1, a_1=1+2\gamma d$, and $\bar a=1+\lambda \gamma d$ thus 

\begin{align*}
 \Tr(\Gamma_m^{\otimes k}(\rho\circ(\evector\evector^*)^{\otimes k}))\geq \delta^{k}
 \left(\frac{1+2\gamma d}{1+\gamma d(2-\lambda)} \right)^{k}.
\end{align*}
and in turn
\begin{align*}
 \MADV_{1-\delta^k}(\sigma^{\otimes k})\geq
 \frac{k\ln (\delta \tfrac{1+2\gamma d}{1+\gamma d(2-\lambda)})}{\ln (1+\gamma)}.
\end{align*}
\end{proof}

We then obtain the following strong direct product theorem for the quantum query complexity of any function (boolean or not).
\begin{theorem}
\label{thm:SDPT-multiplicative}
For any function $f$, any $(2/3)\leq\delta\leq 1$ and any integer $k > 0$, we have
\begin{align*}
Q_{1-\delta^{k/2}}(f^{(k)}) \geq \frac{k\ln(3\delta/2)}{8} \blah(F).
\end{align*}
\end{theorem}

\begin{proof}
Notice that $(J-F)\circ(J-F)=J-F$ and $\gamma_2(J-F)\leq 2$. Thus applying \thmref{thm:product_sigma} with $\lambda=1$ and $\gamma=1/d$, we obtain 
\[
Q^c_{1-\delta^k}(F^{\otimes k}) \ge \frac{k\ln(3\delta/2)}{4} \blah(F).
\]
This lower bound is for computing $f^{(k)}$ coherently, and we obtain the lower bound for $f^{(k)}$ using \claimref{claim:coherently}.
\end{proof}

\section{Boolean functions}

\subsection{XOR Lemma}
We now focus on boolean functions.
Before proving the XOR lemma, we prove a strong direct product theorem for the problem of computing a function in the phase.

Let $\sigma_f=2F-J$ be the Gram matrix corresponding to computing a boolean function $f$ in the phase.
\begin{claim}\label{claim:sdpt-phase-boolean}
 Let $d=\blah(F)$. For any $\delta,\gamma$,
\[
Q^c_{1 - \delta^k}(\sigma_f^{\otimes k}) \ge \frac{k\ln (\delta(1+2\gamma d))}{2\ln (1+\gamma)} \enspace .
\]
\end{claim}

\begin{proof}
Notice that $J-\sigma_f=2(J-F)$, therefore $(J-\sigma_f)\circ(J-\sigma_f)=2(J-\sigma_f)$, $\gamma_2(J-\sigma_f)=2$ and $\blah(\sigma_f)=2\blah(F)$. The claim then follows from \thmref{thm:product_sigma}
with $\lambda=2$.
\end{proof}

Setting $\gamma=1/(\delta d)$, we immediately obtain the strong direct product theorem for $\sigma_f$.
\begin{corollary}
For any $\delta$,
\[
Q^c_{1 - \delta^k}(\sigma_f^{\otimes k}) \ge \frac{k\delta}{4} \blah(F) \enspace .
\]
\end{corollary}

Let $\oplus \circ f^{(k)}$ be the function computing the parity of $k$ independent copies of $f$. Since computing $\oplus \circ f^{(k)}$ in the phase is the same as generating the state $\sigma_f^{\otimes k}$, we obtain the XOR lemma from the strong direct product theorem for $\sigma_f$ 
and~\claimref{claim:computing-phase}.
\begin{corollary}[XOR Lemma]
For any boolean function $f$, any $0\leq\delta\leq 1$ and any integer $k>0$,
\[
Q_{(1 - \delta^{k/2})/2}(\oplus \circ f^{(k)}) \ge \frac{k\delta}{8} \blah(F) \enspace .
\]
\end{corollary}

\subsection{Threshold and strong direct product theorems}
Finally, we prove a threshold direct product theorem.  This will follow from 
\claimref{claim:sdpt-phase-boolean} together with the following threshold 
lemma~\cite[Lemma 2]{Unger09}.

\begin{lemma}[\cite{Unger09}]\label{lem:threshold-lemma}
Let $Y_1, \ldots, Y_k\in\{-1,+1\}$ be random variables, $-1 \le \beta \le 1$ and $C > 0$ be such that 
\[
\Ex\left[\prod_{i \in S} Y_i \right] \le C \beta^{|S|}
\]
for all $S \subseteq [k]$. Let $\lambda$ be such that $\beta\leq\lambda\leq 1$.
Then
\[
\Pr\left[\sum_{i=1}^k Y_i \ge \lambda k\right] \le C e^{-kD(1/2+\lambda/2 || 1/2 + \beta/2)}.
\]
\end{lemma}

\begin{theorem}
\label{tdpt}
For any function $f$, any $0\leq\delta<1$, any $\mu$ such that $\tfrac{1+\sqrt{\delta}}{2}\leq\mu\leq 1$ and any integers $k,K>0$, let $\P_i(x_1,\ldots,x_k)\in\{-1,1\}$ be the $i$-th output of a $T$-query algorithm for $f^{(k)}$, where
\[
 T\leq\frac{k\delta}{K(1-\delta)}\blah(F),
\]
and let $X=\{i\in [k]:\P_i(x_1,\ldots,x_k)=f(x_i)\}$. Then,
\[
  \Pr\left[\abs{X}\geq \mu k\right]\leq e^{\frac{k}{K}-kD\left(\mu||\frac{1+\sqrt{\delta}}{2}\right)}.
\]
\end{theorem}

\begin{proof}
Let $d=\blah(F)$ and, for any $i\in[k]$ and any set $S\subseteq[k]$, let us consider the random variables $Y_i=\P_i(x_1,\ldots,x_k)\cdot f(x_i)\in\{-1,1\}$ and the expectations $\beta_S=E(\prod_{i\in S} Y_i)$. By definition, we have
\begin{align*}
 Q_{(1-\beta_S)/2}(\oplus \circ f^{(\abs{S})})\leq T.
\end{align*}
Moreover, we also have from Claims~\ref{claim:computing-phase} and~\ref{claim:sdpt-phase-boolean}:
\begin{align*}
 Q_{(1-\beta_S)/2}(\oplus \circ f^{(\abs{S})})&\geq \frac{1}{2}Q^c_{1-\beta_S^2}(\sigma_f^{\otimes\abs{S}})
\geq \frac{\ln(\beta_S^2(1+2\gamma d)^{\abs{S}})}{4\ln(1+\gamma)}
\end{align*}
for any $\gamma>0$, which together with the previous inequality leads to
\begin{align*}
 \beta_S\leq (1+\gamma)^{2T}(1+2\gamma d)^{-\abs{S}/2}.
\end{align*}
For $\gamma=(1-\delta)/(2\delta d)$, this implies $\beta_S\leq e^{k/K}\delta^{\abs{S}/2}$.
Using Lemma~\ref{lem:threshold-lemma} with $\beta=\sqrt{\delta}$, $C=e^{k/K}$ and $\lambda=2\mu-1$, we then obtain
\begin{align*}
 \Pr\left[\sum_{i=1}^k Y_i\geq \lambda k\right]\leq e^{\frac{k}{K}-kD\left(\frac{1+\lambda}{2}||\frac{1+\sqrt{\delta}}{2}\right)}.
\end{align*}
The theorem then follows from $\abs{X}=(k+\sum_{i=1}^k Y_i)/2$.
\end{proof}

In the special case $\mu=1$, we obtain the following strong direct product theorem for boolean functions.
\begin{corollary}
For any function $f$, any $0\leq\delta<1$ and any integers $k,K>0$,
\[
Q_{1 - (e^{1/K}(1+\sqrt{\delta})/2)^k}(f^{(k)}) \ge \frac{k\delta}{K(1-\delta)} \blah(F) \enspace .
\]
\end{corollary}

\section*{Acknowledgments}
JR acknowledges support by ARO/NSA under grant W911NF-09-1-0569.
TL would like to thank Ben Reichardt for many insightful conversations on these topics. The authors also thank Oded Regev for interesting comments and in particular for suggesting to prove the XOR lemma. After completion of this work, the authors learned that the quantitative version of the result of Ambainis \textit{et al.}~\cite{AMRR11} about the relation between the multiplicative and additive adversary methods, which was the key missing element to prove the strong direct product theorem, was independently proved by Belovs~\cite{Belovs11}.

\bibliography{sdpt}

\end{document}